\documentclass[authoryear, preprint]{elsarticle}
\usepackage{amsthm}
\usepackage{amsmath}
\usepackage{amsfonts}
\usepackage{amssymb}
\usepackage{graphicx}
\usepackage[english]{babel}
\usepackage{natbib}
\usepackage[margin=1.5in]{geometry}

\newtheorem{prp}{Proposition}

\begin{document}

\title{Multi-Purpose Binomial Model: Fitting all Moments to the Underlying Geometric Brownian Motion}

\author[akim]{Y. S. Kim}
\ead{aaron.kim@stonybrook.edu}

\author[sst]{S. Stoyanov}
\ead{stoyan.stoyanov@stonybrook.edu}

\author[rachev]{S. Rachev}
\ead{svetlozar.rachev@stonybrook.edu}

\author[ff]{F. Fabozzi\corref{ff1}}
\ead{frank.fabozzi@edhec.edu fabozzi321@aol.com}
\cortext[ff1]{Corresponding author: Tel: 01 215 598-8924}

\address[akim]{College of Business, Stony Brook University, Stony Brook, New York 11794-3775}
\address[sst]{College of Business, Stony Brook University, Stony Brook, New York 11794-3775} \address[rachev]{Applied Mathematics and Statistics, Stony Brook University, Stony Brook, New York 11794-3775}
\address[ff]{EDHEC Business School, 393, Promenade des Anglais BP3116, 06202 Nice Cedex 3, Nice, France}

%\author{Aaron Kim
%\\
%{\footnotesize Stony Brook University} \\
%{\footnotesize e-mail: aaron.kim@stonybrook.edu} \\
%\\
%Svetlozar Rachev
%\\
%{\footnotesize Stony Brook University} \\
%{\footnotesize e-mail: svetlozar.rachev@stonybrook.edu} \\
%\\
%Stoyan Stoyanov
%\\
%{\footnotesize Stony Brook University} \\
%{\footnotesize e-mail: stoyan.stoyanov@stonybrook.edu} \\
%\\
%Frank J. Fabozzi\footnote{Corresponding author: Frank J.
%Fabozzi}
%%Preferred address: 858 Tower View Circle, New
%%Hope, PA 18938, tel: +1 (215) 598-8924} 
%\\
%{\footnotesize EDHEC Business School}\\
%{\footnotesize e-mail: fabozzi321@aol.com} }
%\date{}

\begin{abstract}
We construct a binomial tree model fitting all moments to the approximated geometric Brownian motion. Our construction generalizes the classical Cox-Ross-Rubinstein, the Jarrow-Rudd, and the Tian binomial tree models. The new binomial model is used to resolve a discontinuity problem in option pricing. 
\end{abstract}

%\paragraph{Highlight:} 
% \paragraph{JEL Codes: }G12, G13
\maketitle

\newpage
\begin{center}
\begin{Large}
Multi-Purpose Binomial Model: Fitting all Moments to the Underlying Geometric Brownian Motion
\end{Large}
\end{center}
\section{Introduction}
The binomial tree model for option pricing was introduced by \cite{CRR79}, which we refer to in this paper as the CRR model. We assume a stock is traded in the time interval $[0, T]$ and has a price $S_t$ at time $t$. The time interval is divided into $n$ equal periods defined by the instants $t_k = k\Delta t$, $k=0,1, \ldots, n$ in which $n\Delta t=T$. The price at the next time instant $t_k + \Delta t$ equals,\footnote{Here and in what follows, all terms of order $o(\Delta t)$ are omitted.}
\begin{equation*}%\label{eq: Sud}
S_{t_k + \Delta t} = \left\{\begin{array}{ll}S_{t_k + \Delta t, up} = S_{t_k}e^{\sigma\sqrt{\Delta t}} & w.p.\ p_{\Delta t; CRR} = \frac{e^{r\Delta t} - e^{-\sigma\sqrt{\Delta t}}}{e^{\sigma\sqrt{\Delta t}} - e^{-\sigma\sqrt{\Delta t}}}\\ S_{t_k + \Delta t, down} = S_{t_k}e^{-\sigma\sqrt{\Delta t}} & w.p.\ 1 - p_{\Delta t; CRR} \end{array}\right.
\end{equation*}
The underlying discrete stochastic process generated by the tree converges weakly to the geometric Brownian motion (GBM):
\begin{equation}\label{eq: GBM}
S_t = S_{t}^{(r,\sigma)} = e^{(r - \sigma^2/2)t + \sigma B(t)}
\end{equation}
in which $B(t)$ denotes the Brownian motion. 

The weak convergence of a binomial tree to the underlying GBM is understood in the following sense.  Consider the sequence of log-returns, $R_{k\Delta t} = \log(S_{k\Delta t} ) - \log(S_{(k-1)\Delta t}),\ k=1,\ldots,n$. Denote by $R_{t;n},\ t \in [0,T]$  the random polygon linearly connecting the vertexes $(k\Delta t,R_{k\Delta t}),\ k = 0, 1,\ldots, n$. Then verifying the assumptions in  Proposition 3 in \cite{DR08} for a given binomial pricing model, leads to Prokhorov-weak convergence in $C[0,1]$ of the sequence of random polygons to the corresponding arithmetic Brownian motion with instantaneous mean r and instantaneous variance $\sigma^2$. We designate the corresponding  convergence of the binomial tree to the GBM as ``$\stackrel{w}{\rightarrow}$''.

Alternative binomial models were later introduced. The two most notable are the Jarrow-Rudd model and the Tian model. \cite{JR83} match the first two moments of the tree with
\begin{equation*}%\label{eq: GBM}
S_{t_k + \Delta t} = S_{t_k}e^{(r - \sigma^2/2)\Delta t + \sigma(B(t + \Delta t) - B(t))}
\end{equation*}
with a symmetric probability for ``up'' and ``down''
\begin{equation*}%\label{eq: Sud}
S_{t_k + \Delta t} = \left\{\begin{array}{ll}S_{t_k + \Delta t, up} = S_{t_k}e^{(r - \sigma^2/2)\Delta t + \sigma\sqrt{\Delta t}} & w.p.\ p_{\Delta t; JR} = 1/2\\ S_{t_k + \Delta t, down} = S_{t_k}e^{(r - \sigma^2/2)\Delta t - \sigma\sqrt{\Delta t}} & w.p.\ 1 - p_{\Delta t; JR} \end{array}\right.
\end{equation*}
\cite{T93} introduced a tree model 
\begin{equation*}%\label{eq: Sud}
S_{t_k + \Delta t} = \left\{\begin{array}{ll}S_{t_k + \Delta t, up} =S_{t_k}u& w.p.\ p_{\Delta t; TI}\\ S_{t_k + \Delta t, down} = S_t d& w.p.\ 1 - p_{\Delta t; TI} \end{array}\right.
\end{equation*}
where $u = \frac{1}{2}e^{r\Delta t} V(V + 1 + \sqrt{V^2 + 2V - 3})$, $d = \frac{1}{2}e^{r\Delta t} V(V + 1 - \sqrt{V^2 + 2V - 3})$,  $V = e^{\sigma^2\Delta t}$ and $\ p_{\Delta t; TI} = \frac{e^{r\Delta t - d}}{u - d}$ . Tian's model matches the first three moments of the tree and the underlying GBM. 

In this paper we introduce a multi-purpose binomial model which encompasses the CRR,  Jarrow-Rudd, and Tian models as particular cases. Furthermore, a particular case of our model fits all moments of the tree to the underlying GBM.

\section{A Multi-Purpose Binomial Model}

Consider a GBM with instantaneous mean $b \in \mathbb{R}$ and volatility $\sigma > 0$ and the following binomial tree:
\begin{equation}\label{eq: mpbm exp}
S_{t_k + \Delta t} = \left\{\begin{array}{ll}S_{t_k + \Delta t, up} = S_{t_k}\exp((\gamma - h_u^2/2)\Delta t + h_u\sqrt{\Delta t}) & w.p.\ p_{\Delta t}\\ S_{t_k + \Delta t, down} = S_{t_k}\exp((\delta - h_d^2/2)\Delta t - h_d\sqrt{\Delta t}) & w.p.\ 1 - p_{\Delta t} \end{array}\right.
\end{equation}
where  $h_u = \sigma\sqrt{\frac{1 - p}{p}}$, $h_d = \sigma\sqrt{\frac{p}{1 - p}}$ and the tree parameters satisfy the following assumptions:
\begin{itemize}
\item[i. ] $p_{\Delta t} = g + v\sqrt{\Delta t}$ in which $g \in (0, 1)$ and  $v \in \mathbb{R}$. 
\item[ii. ] $\sigma > 0$ is fixed. 
\item[iii. ] The parameter $b$ satisfies $b = g\gamma + (1 - g)\delta$ in which $\gamma  \in \mathbb{R}$ and $\delta \in \mathbb{R}$. 
\item[iv. ] For given model parameters $g \in (0, 1)$ and  $v \in \mathbb{R}$, the frequency $\Delta t > 0$ is small enough so that $p_{\Delta t} \in (0, 1)$.
\end{itemize}
Ignoring the higher order terms, the asymptotics for small $\Delta t$ are
\begin{equation}\label{eq: mpbm}
S_{t_k + \Delta t} = \left\{\begin{array}{ll}S_{t_k + \Delta t, up} = S_{t_k}\left(1 + \gamma\Delta t + \sqrt{\frac{1 - p_{\Delta t}}{p_{\Delta t}}}\sigma\sqrt{\Delta t}\right) & w.p.\ p_{\Delta t}\\ S_{t_k + \Delta t, down} = S_{t_k}\left(1 + \delta\Delta t - \sqrt{\frac{p_{\Delta t}}{1 - p_{\Delta t}}}\sigma\sqrt{\Delta t}\right) & w.p.\ 1 - p_{\Delta t} \end{array}\right.
\end{equation}

The motivation for this model is the following. In all previous binomial models the parameters governing the up’s and down’s probabilities are chosen to simplify the computational complexity; they are completely determined by the instantaneous mean and variance of the underlying GBM which are estimated from historical return data. It is worth noting that in CRR, Jarrow-Rudd and Tian models, $p_{\Delta t}=1/2+o(\sqrt{\Delta t})$ which corresponds in our model to $g=1/2$. Statistical evidence in high frequency trading shows that the assumption $g=1/2$ is problematic. Here we assume that independent of the existing estimates for the mean return and variance, estimates for the $p_{\Delta t}=g+v\sqrt{\Delta t}$ are also available in different return frequencies. This is done by observing the proportions of positive returns in large number of intraday return frequencies.  Thus, our multi-purpose binomial model has an additional parameter $p_{\Delta t}$, which allows us to resolve the option discontinuity puzzle, described later in this paper.

First, we prove that the binomial model converges to GBM. 
\begin{prp}\label{prp: conv} The binomial tree given by \eqref{eq: mpbm} converges in $\stackrel{w}{\rightarrow}$ to \eqref{eq: GBM} with $r = b$. The rate of convergence in terms of the Kolmogorov metric can be approximated by the following relationship
$$\sup_{x \geq 0}|F_{n, S_t}(x) - F_{S_t}(x)| \sim \frac{1}{\sqrt{n}}\times\frac{1 - 2p + 2p^2}{\sqrt{p(1 - p)}}$$
where $F_{n, S_t}(x)$ is the distribution function of $S_t$ approximated by the tree in \eqref{eq: mpbm} with $n$ steps, $F_{S_t}(x)$ is the distribution function of the geometric Brownian motion. 
\end{prp}
\begin{proof} Let $Q_{t_k + \Delta t}^{p_{\Delta t}} = \frac{S_{t_k}}{S_{t_{k-1}}}$ and recall that for $S_{\Delta t} = e^{(b - \sigma^{2}/2)\Delta t + \sigma B(\Delta t)}$, 
\begin{equation}\label{eq: mom1}
ES_{\Delta t}^{n} = e^{n(b + \frac{n-1}{2}\sigma^{2})\Delta t} = 1 + n\left(b + \frac{n - 1}{2}\sigma^2\right)\Delta t. 
\end{equation}
Ignoring the terms of order higher than $\Delta t$, for the first moment we obtain
\begin{align*}
EQ_{t_k + \Delta t}^{p_{\Delta t}} &= \left(1 + \gamma\Delta t + \sqrt{\frac{1 - p_{\Delta t}}{p_{\Delta t}}}\sigma\sqrt{\Delta t}\right)p_{\Delta t} \\ &\qquad+ \left(1 + \delta\Delta t - \sqrt{\frac{p_{\Delta t}}{1 - p_{\Delta t}}}\sigma\sqrt{\Delta t}\right)(1 - p_{\Delta t}) \\&= \left(1 + \gamma\Delta t\right)p_{\Delta t} + \left(1 + \delta\Delta t\right)(1 - p_{\Delta t})\\
&= 1 + (g\gamma + (1 - g)\delta)\Delta t\\
&= 1 + b\Delta t
\end{align*}
A similar calculation shows that $var(Q_{t_k + \Delta t}^{p_{\Delta t}}) = \sigma^2\Delta t$. Finally, Proposition 3 in \cite{DR08} is easily adapted to our case, which proves Proposition \ref{prp: conv}. 

%\begin{align*}
%E(Q_{t_k + \Delta t}^{p_{\Delta t}})^2 &= \left(1 + \gamma\Delta t + \sqrt{\frac{1 - p_{\Delta t}}{p_{\Delta t}}}\sigma\sqrt{\Delta t}\right)^2p_{\Delta t} \\ &\qquad+ \left(1 + \delta\Delta t - \sqrt{\frac{p_{\Delta t}}{1 - p_{\Delta t}}}\sigma\sqrt{\Delta t}\right)^2(1 - p_{\Delta t}) \\ &= \left(1 + \left(2\gamma + \frac{1 - p_{\Delta t}}{p_{\Delta t}}\sigma^2\right)\Delta t + 2\sqrt{\frac{1 - p_{\Delta t}}{p_{\Delta t}}}\sigma\sqrt{\Delta t}\right)p_{\Delta t} \\ &\qquad+ \left(1 + \left(2\delta + \frac{ p_{\Delta t}}{1 - p_{\Delta t}}\sigma^2\right)\Delta t  - 2\sqrt{\frac{p_{\Delta t}}{1 - p_{\Delta t}}}\sigma\sqrt{\Delta t}\right)(1 - p_{\Delta t}) \\ &= \left(1 + \left(2\gamma + \frac{1 - p_{\Delta t}}{p_{\Delta t}}\sigma^2\right)\Delta t\right)p_{\Delta t} + \left(1 + \left(2\delta + \frac{ p_{\Delta t}}{1 - p_{\Delta t}}\sigma^2\right)\Delta t\right)(1 - p_{\Delta t}) \\
%&= 1 + (2\gamma p_{\Delta t} + 2\delta(1 - p_{\Delta t}) + \sigma^{2})\Delta t\\
%&= 1 + (2\gamma g+ 2\delta(1 -g) + \sigma^{2})\Delta t\\
%&= 1 + 2b\Delta t + \sigma^2\Delta t \\
%&= (EQ_{t_k + \Delta t}^{p_{\Delta t}})^2 + \sigma^2\Delta t 
%\end{align*}
%Because $var(Q_{t_k + \Delta t}^{p_{\Delta t}}) = E(Q_{t_k + \Delta t}^{p_{\Delta t}})^2 - (EQ_{t_k + \Delta t}^{p_{\Delta t}})^2$, it follows that $var(Q_{t_k + \Delta t}^{p_{\Delta t}}) = \sigma^2\Delta t$. 
The rate of convergence follows from the functional form in \eqref{eq: mpbm exp}. Without loss of generality, assume $t = 1$. The price at $t = 1$ can be represented as
\begin{equation*}%\label{eq: mpbm rv}
S_{1,n} = \prod_{k = 1}^{n} e^{(V_k - X_k^2/2)\frac{1}{n} + X_k\frac{1}{\sqrt{n}}}
= \exp\left(\frac{1}{n}\sum_{k = 1}^{n}V_k - \frac{1}{2n}\sum_{k = 1}^{n} X_k^2 + \frac{1}{\sqrt{n}}\sum_{k = 1}^{n}X_k\right)
\end{equation*}
where $(V_k, X_k)$ are independent for $k = 1,2,\ldots$ taking values $(\gamma, \sigma\sqrt{\frac{1 -  p}{p}})$ with probability $p_{\Delta t}$ and $(\delta, -\sigma\sqrt{\frac{1}{1 - p}})$ with probability $1 - p_{\Delta t}$. By the strong law of large numbers, $\frac{1}{n}\sum_{k = 1}^{n}V_k \rightarrow b$ and $\frac{1}{n}\sum_{k = 1}^{n} X_k^2 \rightarrow \sigma^2$ almost surely. Denote $Y_n = \frac{1}{\sqrt{n}}\sum_{k = 1}^{n}X_k$, $S_{1,n} = e^{b - \sigma^2/2}U_n$ in which $U_n = e^{Y_n}$, and $Z \in N(0, \sigma^2)$. Ignoring the impact of the constants, by the Berry-Esseen theorem\footnote{See, for example, \cite[pp 542]{Fel71}.} we obtain the following bound,
\begin{align*}%\label{eq: mpbm rv}
\sup_{x \geq 0}|F_{U_n}(x) - F_{e^Z}(x)| &= \sup_{x \geq 0}|P(Y_n \leq \log{x}) - P(Z \leq \log{x})|\\
&= \sup_{y \in \mathbb{R}}|P(Y_n \leq y) - P(Z \leq y)|\\
&\leq C\frac{E|X_1|^3}{\sigma^3\sqrt{n}} = \frac{C}{\sqrt{n}}\times\frac{1 - 2p + 2p^2}{\sqrt{p(1 - p)}}
\end{align*}
where $C$ is an absolute constant. 
\end{proof}

Next, we demonstrate that the CRR, Jarrow-Rudd, and Tian models arise from \eqref{eq: mpbm}. 
\begin{prp} The following statements hold:
\begin{itemize}
\item[a) ] We obtain the CRR model by setting $\gamma = \delta  = r$, $g = 1/2$, and $v = \frac{r - \sigma^2/2}{2\sigma}$ in \eqref{eq: mpbm}. 
\item[b) ] We obtain the Jarrow-Rudd model by setting $\gamma = \delta = r$, $g = 1/2$, and $v = 0$ in \eqref{eq: mpbm}. 
\item[c) ] We obtain the Tian model by setting $\gamma = \delta = r + \frac{3}{2}\sigma^2$, $g = 1/2$, and $v = \frac{r - \sigma^2/2}{2\sigma}$ in \eqref{eq: mpbm}. 
\end{itemize}
\end{prp}
\begin{proof}
The proof is done by comparing the leading terms and is straightforward. 
\end{proof}

The next proposition generalizes Tian's model by constructing a binomial model which fits all tree moments to the corresponding moments of the underlying GBM. The result is a more robust binomial tree model, as it matches not only the first tree central moments but also the kurtoses of the tree and the underlying log-normal return distribution.

\begin{prp}\label{prp: Tree mom} Setting in the binomial tree model  in \eqref{eq: mpbm} $\gamma  = \delta =  b$ and $v = 0$ implies $E(Q_{t_k + \Delta t}^{p_{\Delta t}} )^j = ES_{\Delta t}^{j}$ for all $j = 1, 2, \ldots$. 
\end{prp}
\begin{proof}
The result can be viewed as a new version of the well-known moment problem for the log-normal distribution, see \cite{H63}: the lognormal distribution is not determined by the set of its moments. Applying (8) and inductively calculating, 
\begin{align*}
E(Q_{t_k + \Delta t}^{p_{\Delta t}} )^j &= \left(1 + b\Delta t + \sqrt{\frac{1 - p_{\Delta t}}{p_{\Delta t}}}\sigma\sqrt{\Delta t} \right)^jp_{\Delta t} \\ &\qquad + \left(1 + b\Delta t - \sqrt{\frac{p_{\Delta t}}{1 - p_{\Delta t}}}\sigma\sqrt{\Delta t} \right)^j(1 - p_{\Delta t}) \\
&= \left(1 + \left(jb + \frac{j(j-1)}{2}\frac{1 - p_{\Delta t}}{p_{\Delta t}}\sigma^2\right)\Delta t + j\sqrt{\frac{1 - p_{\Delta t}}{p_{\Delta t}}}\sigma\sqrt{\Delta t}\right)p_{\Delta t}\\
&= 1 + \left(jb + \frac{j(j-1)}{2}\sigma^2\right)\Delta t
\end{align*}
proves the proposition. 
\end{proof}

\section{On the Problem of Continuity in Option Pricing}

In this section we make use of our multi-purpose binomial model to resolve a discontinuity problem in option pricing. We first describe the problem for the CRR binomial model and then provide the solution. 

\subsection{The Discontinuity Problem in Option Pricing Models}

To illustrate the discontinuity problem in the CRR model,  consider a one-step binomial model, where (1) $S_0=1$ is the current (at $t=0$)  (one-share ) stock price; (2) $f_0$  is the unknown current option price; and (3) $S_T = u = e^{\sigma\sqrt{T}}>1$ (resp. $d=1/u$) with probability $p$  (resp. $1-p$).  The option payoff at maturity $T$ is
\begin{equation*}
f_T = \left\{\begin{array}{ll}f_u = g(S_0u) & w.p.\ p \\ f_d = g(S_0d) & w.p.\ 1 - p\end{array}\right.
\end{equation*}
for some option payoff function $g(x)$, $x \in \mathbb{R}$. Indeed, when $p \in (0, 1)$ the value of the option at $t = 0$ is given by $f_0(p) = f_0(1/2) = e^{-rT}(\tilde q f_{u} + (1 - \tilde q)f_{d})$ with $\tilde q = \frac{e^{rT} - d}{u - d}$ regardless of how close $p$ is to 0 or 1. However, for $p = 0$ or $p = 1$, the option values are $f_0(0) = e^{-rT}f_d$ and $f_0(1) = e^{-rT}f_u$. The discontinuity at  $p = 0$ and $p = 1$
\begin{align*}
f_{0}(0) - \lim_{p \rightarrow 0^+}f_0(p) &= e^{-rT}\tilde q (f_d - f_u)\\
f_{0}(1) - \lim_{p \rightarrow 1^-}f_0(p) &= e^{-rT}(1 - \tilde q)(f_u - f_d)
\end{align*}
seems unnatural. In continuous time, this problem translates to no sensitivity of the option price to the mean of the GBM which is unreasonable in a non-Gaussian setting, see \cite[Chapter 4.5]{BP00}. 

\subsection{Resolving the Discontinuity Problem}

Suppose that returns of frequency $\Delta t$ are collected and the following quantities are estimated from that data: (1) $p_{\Delta t}$, (2) the upward movement in $\Delta t$, $1 + \gamma\Delta t + \sqrt{\frac{1 - p_{\Delta t}}{p_{\Delta t}}}\sigma\sqrt{\Delta t}$, (3) the downward movement in $\Delta t$, $1 + \delta\Delta t + \sqrt{\frac{1 - p_{\Delta t}}{p_{\Delta t}}}\sigma\sqrt{\Delta t}$, and (4) the volatility $\sigma$.

Having estimated all parameters, we set up to calculate the risk-neutral probabilities. We construct a portfolio of the derivative instrument and a delta $\Delta_{t_k}:= \Delta^{(p_{\Delta t})}_{t_k}$ position in the underlying stock at time $t_k$. The payoff of the portfolio $P^{t_{k + 1}} = P^{t_k + \Delta t}$ one time step ahead equals
\begin{align*}
P^{t_{k + 1}} = \left\{\begin{array}{ll}\Delta_{t_k}S_{t_k}u - f_{u, t_{k+1}} = \Delta_{t_k}S_{t_k}\left(1 + \gamma \Delta t + \sqrt{\frac{ 1 -p_{\Delta t} }{p_{\Delta t}}}\sigma\sqrt{\Delta t}\right) - f_{u, t_{k + 1}} & w.p.\ p_{\Delta t} \\ \Delta_{t_k}S_{t_k}d - f_{d, t_{k + 1}} = \Delta_{t_k}S_{t_k}\left(1 + \delta \Delta t - \sqrt{\frac{ p_{\Delta t} }{1 - p_{\Delta t}}}\sigma\sqrt{\Delta t}\right) - f_{d, t_{k + 1}} & w.p.\ 1 - p_{\Delta t}\end{array} \right.
\end{align*}
To compute the delta, we set the variance of the portfolio equal to zero, $var(P^{t_{k + 1}}) = 0$. The value of delta obtained in this way equals
$$\Delta_{t_k} = \frac{1}{S_{t_k}}\times\frac{f_{u,t_{k+1}} - f_{d,t_{k + 1}}}{(\gamma - \delta)\Delta t + \frac{\sigma\sqrt{\Delta t}}{\sqrt{p_{\Delta t}(1 - p_{\Delta t})}}}$$
Because the portfolio constructed in this way is riskless, the price of the portfolio at $t_k$ equals the present value of the future certain cashflow, $\Delta_{t_k}S_{t_k} - f_{t_k} = e^{-r\Delta t}E(P^{t_k + \Delta t})$. Thus, for the price of the derivative we obtain
$$f_{t_k} = e^{-r\Delta t}(Q_{\Delta t}f_{u, t_{k + 1}} + (1 - Q_{\Delta t})f_{d, t_{k+1}})$$
in which the risk-neutral probability equals
$$Q_{\Delta t} = \frac{(r - \delta) \sqrt{p_{\Delta t}(1 - p_{\Delta t})}\sqrt{\Delta t} + p_{\Delta t}\sigma}{(\gamma - \delta) \sqrt{p_{\Delta t}(1 - p_{\Delta t})}\sqrt{\Delta t} + \sigma}$$
In the particular case of $\gamma = \delta$, $Q_{\Delta t} = p_{\Delta t} - \theta\sqrt{p_{\Delta t}(1 - p_{\Delta t})}$ in which $\theta = \frac{\gamma - r}{\sigma}$ denotes the market price for risk.  The risk-neutral probability $Q_{\Delta t}$ is continuous in $p_{\Delta t}$ approaching the limits 0 and 1 which resolves the discontinuity problem. 

\section{Estimating $p_{\Delta t}$ under the Physical Measure}

In this section, we provide empirical evidence that (\textit{i}) $p$ is generally not equal to $1/2$ and (\textit{ii}) $p$ is not constant through time. We take the daily returns of the S\&P 500 index from January 3, 1950 to May 20, 2016 and we consider the indicator of the event ``market is up'' on a daily basis. The total number of observations is 16,703. The number of days with a positive daily return is 8,836 and $\hat p = 0.529$ with a 95\% confidence interval of $[0.5214, 0.5366]$ which rejects the hypothesis $H_0: p = 1/2$. 

\begin{figure}
\centering
\includegraphics[scale=0.6]{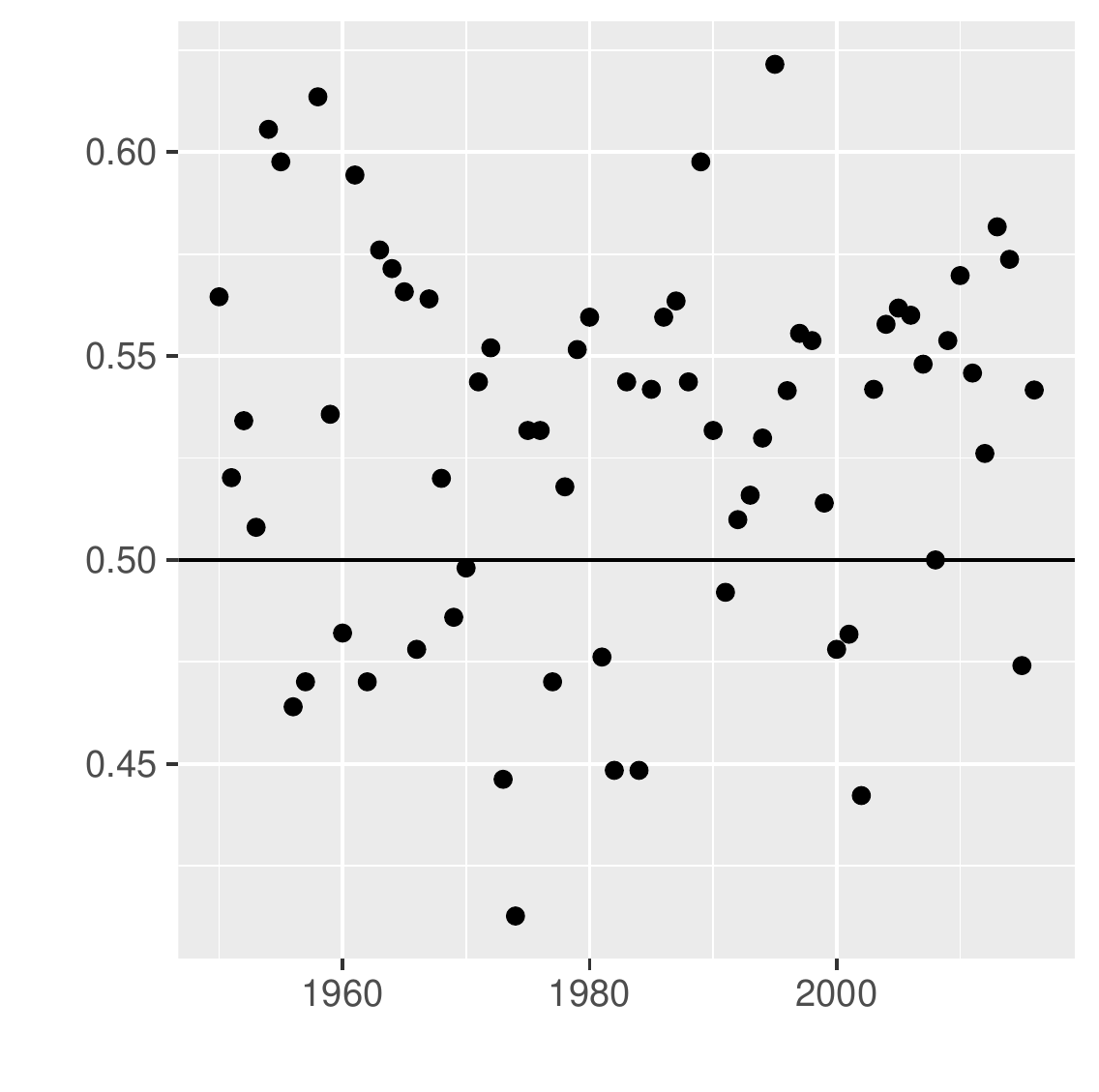}
\caption{Estimated probabilities $p$ for a daily upward move of S\&P 500 for each year from 1950 to 2016. }\label{fig: prob}
\end{figure}

We also test if we can accept the hypothesis that $p$ is the same if estimated for each year in the sample, $H_0: p_1 = p_2 = \ldots = p_k$ where $k$ equals the number of years in the sample against the alternative that it is different in at least one year.\footnote{We use the exact binomial test and Newcombe's test as implemented in binom.test, and prop.test in the standard statistics library in R. } The $p$-value of the test is $4.14\times 10^{-7}$ which strongly rejects $H_0$. The estimated probabilities are plotted in Figure \ref{fig: prob}. 

\section{Calibrating $p_{\Delta t}$ under the Risk Neutral Measure}

To estimate model parameters under the risk neutral measure, we consider the following calibration problem. First, we select a set of call options written on S\&P 500 which  mature in less than 100 working days. The model parameters are calibrated on August 7, August 8, September 10, September 15 (the date of Lehman Brothers Collapse), and September 16, 2008. We select the 13-week Treasury bond index (IRX) as the risk-free asset. Also we set the number of time steps $n$ as the number of days to maturity and $\varDelta t = 1/252$.

Second, we calibrate the parameters of the binomial tree models by minimizing the error between the theoretical prices and the observed closing market prices on a particular day as measured by the sum of squared differences (RMSE). Several other measures of the error are also computed: the average absolute error (AAE), the average absolute error as a percentage of the mean price (APE), and the average relative percentage error (ARPE), defined as follows:
\begin{align*}
&    \textup{AAE} =
    \sum_{j=1}^N \frac{|{P}_j -
    \widehat{P}_j|}{N},
~~    \textup{APE} = \frac{\sum_{j=1}^N \frac{|{P}_j -
    \widehat{P}_i|}{N}}{\sum_{j=1}^N\frac{{P}_j}{N}},
    \\
&
    \textup{ARPE} = \frac{1}{N}
    \sum_{j=1}^N \frac{|{P}_j -
    \widehat{P}_j|}{{P}_j},
~~    \textup{RMSE} = \sqrt{
    \sum_{j=1}^N \frac{({P}_j -
    \widehat{P}_j)^2}{N}},
\end{align*}
where $\widehat{P}_j$ and ${P}_j$ are model prices and observed market prices of call options with strikes $K_j$, time to maturity $T_j$,  $j\in\{1,\ldots,N\}$, and $N$ is the number of observed call option prices.

We consider two cases of the multi-propose binomial tree. In the first case (MP-Bin1), we use the setting of Proposition \ref{prp: Tree mom}. That is, we assume $\gamma=\delta=r$ and $v=0$, and estimate the parameters $\sigma$ and $p_{\varDelta t} = g$. In the second case (MP-Bin2), we estimate the parameters $\gamma, g, p_{\varDelta t}\in(0,1)$, and $\sigma>0$ by means of the least squares curve fit with $v=(p_{\varDelta t}-g)/\sqrt{\varDelta t}$ and $\delta = (r - g\gamma)/(1-g)$. 

To compare the performance to the other binomial trees, we calibrate $\sigma$ of CRR, Jarrow-Rudd and Tian models using the same method. The calibrated parameters are provided in Table \ref{Table:CalibrationRiskNeutralMeasure} and the four types of error for the calibration problem are presented in Table \ref{Table:ErrorEstimator}. 

The MP-Bin2 model has the smallest error for each parameter estimation and the MP-Bin1 model has the second smallest error for each parameter estimation. Finally, the tables show that the calibrated probability $p_{\Delta t}$ can be quite different from $1/2$ providing a better fit to the observed option prices. 

\begin{table}[t]
%\vspace{-2cm}
\begin{center}
%\begin{footnotesize}
\begin{tabular}{cccc}%{c@{ }|@{ }c@{ }|@{}cccccccccc}
\hline
Date & Model & \multicolumn{2}{c}{Parameters} \\
\hline
Aug 06,2008 & CRR & $\sigma = 0.1772$ &  $p_{\varDelta t} = 0.5001$ \\ 
           & Jarrow-Rudd &$\sigma = 0.1773$&  $p_{\varDelta t} = 1/2$ \\ 
           & Tian & $\sigma = 0.1766$&  $p_{\varDelta t} = 0.4917$ \\ 
           &MP-Bin1 & $\sigma = 0.1776$&  $p_{\varDelta t} = 0.5409$ \\ 
           &MP-Bin2 & $\sigma = 0.1549$&  $p_{\varDelta t} = 0.6259$ \\ 
       &          & \multicolumn{2}{c}{($\gamma = 0.2267$, $\delta = 0.0160$, $g = 0.0001$)} \\ 
\hline
Aug 07,2008 & CRR & $\sigma = 0.1875$ &  $p_{\varDelta t} = 0.4998$ \\ 
           & Jarrow-Rudd &$\sigma = 0.1875$&  $p_{\varDelta t} = 1/2$ \\ 
           & Tian & $\sigma = 0.1870$&  $p_{\varDelta t} = 0.4912$ \\ 
           &MP-Bin1 & $\sigma = 0.1874$&  $p_{\varDelta t} = 0.5051$ \\ 
           &MP-Bin2 & $\sigma = 0.1756$&  $p_{\varDelta t} = 0.5404$ \\ 
       &          & \multicolumn{2}{c}{($\gamma = 0.1324$, $\delta = 0.0163$, $g = 0.0001$)} \\ 
\hline
Sep 03,2008 & CRR & $\sigma = 0.1977$ &  $p_{\varDelta t} = 0.4995$ \\ 
           & Jarrow-Rudd &$\sigma = 0.1977$&  $p_{\varDelta t} = 1/2$ \\ 
           & Tian & $\sigma = 0.1975$&  $p_{\varDelta t} = 0.4907$ \\ 
           &MP-Bin1 & $\sigma = 0.1975$&  $p_{\varDelta t} = 0.6562$ \\ 
           &MP-Bin2 & $\sigma = 0.1847$&  $p_{\varDelta t} = 0.6135$ \\ 
       &          & \multicolumn{2}{c}{($\gamma = 0.1713$, $\delta = 0.0166$, $g = 0.0001$)} \\ 
\hline
Sep 10,2008 & CRR & $\sigma = 0.2049$ &  $p_{\varDelta t} = 0.4993$ \\ 
           & Jarrow-Rudd &$\sigma = 0.2048$&  $p_{\varDelta t} = 1/2$ \\ 
           & Tian & $\sigma = 0.2049$&  $p_{\varDelta t} = 0.4903$ \\ 
           &MP-Bin1 & $\sigma = 0.2049$&  $p_{\varDelta t} = 0.5215$ \\ 
           &MP-Bin2 & $\sigma = 0.2188$&  $p_{\varDelta t} = 0.5478$ \\ 
       &          & \multicolumn{2}{c}{($\gamma = 0.0001$, $\delta = 0.1742$, $g = 0.9075$)} \\ 
\hline
Sep 15,2008 & CRR & $\sigma = 0.2548$ &  $p_{\varDelta t} = 0.4970$ \\ 
           & Jarrow-Rudd &$\sigma = 0.2538$&  $p_{\varDelta t} = 1/2$ \\ 
           & Tian & $\sigma = 0.2571$&  $p_{\varDelta t} = 0.4879$ \\ 
           &MP-Bin1 & $\sigma = 0.2523$&  $p_{\varDelta t} = 0.5462$ \\ 
           &MP-Bin2 & $\sigma = 0.3371$&  $p_{\varDelta t} = 0.6956$ \\ 
       &          & \multicolumn{2}{c}{($\gamma = 0.0001$, $\delta = 0.9536$, $g = 0.9916$)} \\ 
\hline
Sep 16,2008 & CRR & $\sigma = 0.2412$ &  $p_{\varDelta t} = 0.4973$ \\ 
           & Jarrow-Rudd &$\sigma = 0.2411$&  $p_{\varDelta t} = 1/2$ \\ 
           & Tian & $\sigma = 0.2418$&  $p_{\varDelta t} = 0.4886$ \\ 
           &MP-Bin1 & $\sigma = 0.2433$&  $p_{\varDelta t} = 0.5715$ \\ 
           &MP-Bin2 & $\sigma = 0.3084$&  $p_{\varDelta t} = 0.6545$ \\ 
       &          & \multicolumn{2}{c}{($\gamma = 0.0001$, $\delta = 0.6275$, $g = 0.9865$)} \\ 
\hline
\end{tabular}
%\end{footnotesize}
%{\\ $N$ is the number of observed call options for the calibration.} 
\caption{\label{Table:CalibrationRiskNeutralMeasure} Calibrated parameter values using observed option prices. }
\end{center}
\end{table}           
           
\begin{table}
\begin{center}
%\begin{footnotesize}
\begin{tabular}{cccccc}
\hline
Date & Model & AAE & APE & ARPE & RMSE \\
\hline
Aug 06,2008 &CRR & $1.8061$ & $0.0703$ & $0.4249$ & $2.2644$ \\ 
           & Jarrow-Rudd & $1.8081$ & $0.0704$ & $0.4265$ & $2.2660$ \\ 
           & Tian & $1.8519$ & $0.0721$ & $0.4295$ & $2.3808$ \\ 
           & MP-Bin1 & $1.6236$ & $0.0632$ & $0.3897$ & $1.9956$ \\ 
           & MP-Bin2 & $\textbf{1.4046}$ & $\textbf{0.0547}$ & $\textbf{0.2536}$ & $\textbf{1.7553}$ \\ 
\hline
Aug 07,2008 &CRR & $1.8530$ & $0.0667$ & $0.4421$ & $2.2840$ \\ 
           & Jarrow-Rudd & $1.8506$ & $0.0666$ & $0.4413$ & $2.2829$ \\ 
           & Tian & $1.8700$ & $0.0673$ & $0.4691$ & $2.2412$ \\ 
           & MP-Bin1 & $1.8312$ & $0.0659$ & $0.4334$ & $2.2777$ \\ 
           & MP-Bin2 & $\textbf{1.7786}$ & $\textbf{0.0640}$ & $\textbf{0.3868}$ & $\textbf{2.2404}$ \\ 
\hline
Sep 03,2008 &CRR & $2.2132$ & $0.0310$ & $0.3581$ & $2.6879$ \\ 
           & Jarrow-Rudd & $2.2091$ & $0.0309$ & $0.3577$ & $2.6840$ \\ 
           & Tian & $2.2433$ & $0.0314$ & $0.3713$ & $2.7282$ \\ 
           & MP-Bin1 & $1.7327$ & $0.0243$ & $0.2399$ & $2.1267$ \\ 
           & MP-Bin2 & $\textbf{1.9424}$ & $\textbf{0.0272}$ & $\textbf{0.2544}$ & $\textbf{2.3893}$ \\ 
\hline
Sep 10,2008 &CRR & $1.2890$ & $0.0749$ & $0.3397$ & $1.7425$ \\ 
           & Jarrow-Rudd & $1.2931$ & $0.0752$ & $0.3377$ & $1.7409$ \\ 
           & Tian & $1.2876$ & $0.0748$ & $0.3513$ & $1.8472$ \\ 
           & MP-Bin1 & $1.2859$ & $0.0748$ & $0.3177$ & $1.7016$ \\ 
           & MP-Bin2 & $\textbf{1.2328}$ & $\textbf{0.0717}$ & $\textbf{0.3258}$ & $\textbf{1.6401}$ \\ 
\hline
Sep 15,2008 &CRR & $2.5501$ & $0.1540$ & $0.4461$ & $3.4683$ \\ 
           & Jarrow-Rudd & $2.4660$ & $0.1489$ & $0.4394$ & $3.3530$ \\ 
           & Tian & $2.6321$ & $0.1589$ & $0.4560$ & $3.6495$ \\ 
           & MP-Bin1 & $2.3143$ & $0.1397$ & $0.4483$ & $3.1775$ \\ 
           & MP-Bin2 & $\textbf{1.6918}$ & $\textbf{0.1022}$ & $\textbf{0.4429}$ & $\textbf{2.2074}$ \\ 
\hline
Sep 16,2008 &CRR & $2.7063$ & $0.0895$ & $0.4272$ & $3.4906$ \\ 
           & Jarrow-Rudd & $2.6631$ & $0.0881$ & $0.4250$ & $3.4441$ \\ 
           & Tian & $2.9515$ & $0.0976$ & $0.4638$ & $3.7932$ \\ 
           & MP-Bin1 & $2.4370$ & $0.0806$ & $0.3896$ & $3.2011$ \\ 
           & MP-Bin2 & $\textbf{1.9250}$ & $\textbf{0.0637}$ & $\textbf{0.3843}$ & $\textbf{2.6229}$ \\ 
\hline
\end{tabular}
%\end{footnotesize}
%{\\ \footnotesize } 
\caption{\label{Table:ErrorEstimator} The different errors at the calibrated parameter values. }
\end{center}
\end{table}

\section{Conclusion}
We propose a new multi-purpose binomial tree model which generalizes the CRR, the Jarrow-Rudd, and the Tian tree models. Our model is more flexible and can match asymptotically all moments of the limiting log-normal distribution. We apply the model to resolve a discontinuity problem in option pricing when the probability for ``up'' converges to 0 or 1. We also provide an estimation and a calibration example which illustrate that the suggested binomial model is more realistic than the three alternatives. 

% {\footnotesize \bibliography{..//References//AllRefs}}
\section*{Acknowledgments}
The authors are grateful for the helpful comments of an anonymous referee. 
\section*{References}
\bibliography{..//References//AllRefs}

\begin{thebibliography}{7}
\expandafter\ifx\csname natexlab\endcsname\relax\def\natexlab#1{#1}\fi
\expandafter\ifx\csname url\endcsname\relax
  \def\url#1{\texttt{#1}}\fi
\expandafter\ifx\csname urlprefix\endcsname\relax\def\urlprefix{URL }\fi

\bibitem[{Bouchard and Potters(2000)}]{BP00}
Bouchard, J.-P., Potters, M., 2000. Theory of Financial Risks: From Statistical
  Physics to Risk Management. Cambridge University Press.

\bibitem[{Cox et~al.(1979)Cox, Ross, and Rubinstein}]{CRR79}
Cox, J., Ross, S., Rubinstein, M., 1979. Options pricing: a simplified
  approach. Journal of Financial Economics 7, 229--263.

\bibitem[{Davydov and Rotar(2008)}]{DR08}
Davydov, Y., Rotar, V., 2008. On a non-classical invariance principle. Stat.
  Probab. Letters V(78), 2031--2038.

\bibitem[{Feller(1971)}]{Fel71}
Feller, W., 1971. An Introduction to Probability Theory and its Application.
  2nd Edition. Wiley, New York.

\bibitem[{Heyde(1963)}]{H63}
Heyde, C.~C., 1963. On a property of the lognormal distribution. J. Royal
  Statist. Soc., Ser. B 25, 392--393.

\bibitem[{Jarrow and Rudd(1983)}]{JR83}
Jarrow, R., Rudd, A., 1983. Option Pricing. Homewood, IL: Dow Jones-Irwin
  Publishing.

\bibitem[{Tian(1993)}]{T93}
Tian, Y., 1993. A modified lattice approach to option pricing. The Journal of
  Futures Markets 13, 563--577.

\end{thebibliography}
\end{document}